\documentclass[aps,floatfix,letterpaper,twocolumn,longbibliography]{revtex4-1}

\usepackage{graphicx} 

\usepackage{amsmath,amsfonts,amssymb,amscd}
\usepackage{amsthm}
\usepackage{thmtools}
\usepackage{thm-restate}
\usepackage[usenames,dvipsnames]{xcolor}
\usepackage{enumerate}
\usepackage{fancyhdr}

\usepackage{hyperref}
\hypersetup{
    bookmarksnumbered=true, 
    unicode=false, 
    pdffitwindow=true,
    pdfstartview={FitH}, 
    pdftitle={Robust Extraction of Tomographic Information via Randomized Benchmarking}, 
    pdfauthor={}, 
    pdfsubject={}, 
    pdfcreator={}, 
    pdfproducer={}, 
    pdfkeywords={}, 
    pdfnewwindow=true, 
    colorlinks=true, 
    linkcolor=blue, 
    citecolor=blue, 
    urlcolor=magenta 
}

\newtheorem{thm}{Theorem}[section]

\newtheorem{lem}[thm]{Lemma}
\newtheorem{claim}[thm]{Claim}

\numberwithin{equation}{section}

\newcommand{\eq}[1]{\hyperref[eq:#1]{Eq. (\ref*{eq:#1})}}

\renewcommand{\sec}[1]{\hyperref[sec:#1]{Section~\ref*{sec:#1}}}
\newcommand{\thrm}[1]{\hyperref[thm:#1]{Theorem~\ref*{thm:#1}}}
\newcommand{\lemm}[1]{\hyperref[lemm:#1]{Lemma~\ref*{lemm:#1}}}
\newcommand{\prop}[1]{\hyperref[prop:#1]{Proposition~\ref*{prop:#1}}}
\newcommand{\corr}[1]{\hyperref[corr:#1]{Corollary~\ref*{corr:#1}}}
\newcommand{\fig}[1]{\hyperref[fig:#1]{Figure~\ref*{fig:#1}}}
\newcommand{\app}[1]{\hyperref[app:#1]{Appendix~\ref*{app:#1}}}

\usepackage{xcolor}

\DeclareMathAlphabet{\matheu}{U}{eus}{m}{n}

\DeclareMathOperator{\tr}{tr}

\newcommand{\Pauli}{{\matheu{P}}}

\newcommand{\lamCn}{\sop N}

\newcommand{\bra}[1]{\langle{#1}|}
\newcommand{\ket}[1]{|{#1}\rangle}

\newcommand{\ketbra}[2]{|{#1}\rangle\!\langle{#2}|}

\newcommand{\kket}[1]{|{#1}\rangle\!\rangle}

\newcommand{\op}[1]{{\hat{#1}}}
\newcommand{\sop}[1]{{\mathcal{#1}}}
\newcommand{\e}{{\mathrm{e}}}
\newcommand{\Favg}{{\overline{F}}}

\newcommand{\I}{{\op {\mathbb I}}}

\newcommand{\CNOT}{{\mathrm{CNOT}}}

\newcommand{\PL}{{(\mathrm{PL})}}

\begin{document}


\title{Robust Extraction of Tomographic Information via
  Randomized Benchmarking} \author{Shelby Kimmel} 
\affiliation{Center
  for Theoretical Physics, MIT, Cambridge, MA}
\affiliation{Raytheon BBN Technologies, Quantum Information
  Processing Group, Cambridge, MA}
\author{Marcus P. da Silva, Colm A. Ryan, Blake R. Johnson, Thomas
  Ohki} 
\affiliation{Raytheon BBN Technologies, Quantum Information
  Processing Group, Cambridge, MA}

\begin{abstract}
  We describe how randomized benchmarking can be used to reconstruct
  the unital part of any trace-preserving quantum map, which in turn
  is sufficient for the full characterization of any unitary
  evolution, or more generally, any unital trace-preserving
  evolution. This approach inherits randomized benchmarking's
  robustness to preparation, measurement, and gate imperfections,
  therefore avoiding systematic errors caused by these imperfections.
  We also extend these techniques to efficiently estimate the average
  fidelity of a quantum map to unitary maps outside of the Clifford
  group. The unitaries we consider correspond to large circuits
  commonly used as building blocks to achieve scalable, universal, and
  fault-tolerant quantum computation. Hence, we can efficiently verify
  all such subcomponents of a circuit-based universal quantum computer.
  In addition, we rigorously bound the time and sampling complexities
  of randomized benchmarking procedures, proving that the required
  non-linear estimation problem can be solved efficiently.
\end{abstract}

\maketitle


\section{Introduction}
\label{sec:introduction}

While quantum process tomography~\cite{CN97} is a conceptually simple
approach to the characterization of quantum operations on states, its
implementation suffers from a number of fundamental drawbacks. These
obstacles range from its exponential scaling with the size of the
system, to its dependence on precise knowledge of state preparation
and measurement. Precise knowledge about state preparation requires
precise knowledge about operations and measurements, leading to a
difficult non-linear estimation
problem~\cite{stark1,stark2,stark3,MGS+12}. Lack of precise knowledge
about state preparation and measurement can also lead to significant
systematic errors in the reconstructed operations~\cite{WHE+04}.  Recently,
{\em randomized benchmarking} (RB) protocols have been shown to lead
to estimates of the average fidelity to Clifford group operations in a
manner that is robust against imprecise knowledge about state
preparation and measurement, and therefore largely free of some of the
systematic errors that can affect standard tomographic
reconstructions~\cite{EAZ05,KLR+08,MGE11,MGE12,GMT+12,MGJ+12}.

We describe a procedure that provides an almost complete
description of any quantum map in a way that is robust against
many errors that plague standard tomographic procedures. Specifically, 
we can estimate the unital part~\cite{HKL04,BZ08} of any trace-preserving
map, which includes all parameters
necessary to describe deterministic as well as random unitary
evolution. Furthermore, we show that a related protocol can be used to efficiently estimate 
the average fidelity to unitary operations outside the Clifford
group, again in a way that is accurate even in the presence
of state preparation, measurement, and unitary control errors. 

Both procedures use RB protocols as a tool, combined with several
new results: we show that Clifford group maps span
the unital subspace of quantum maps, and that important unitaries outside the Clifford group
 can be expressed as linear combinations of few Clifford group maps.
These insights, combined with new error strategies and analysis, 
allow us to robustly characterize maps that were previously inaccessible.

Our error analysis rigorously
proves that randomized benchmarking decays can be fit efficiently. 
We also prove new results on the average fidelity of composed maps, which
is important for RB, but is also of significance to any procedure where
direct access to a quantum map is limited.

This paper is organized as follows. In \sec{background} we give
background on general properties of quantum operations. In \sec{RB} we
sketch the RB protocol and describe the information that can be
extracted from such experiments. In \sec{tomography} we describe how
the information from RB experiments can be used to tomographically
reconstruct the unital part of any experimental map even in the
presence of imperfect randomizing operations.  In \sec{fidTandU}, we
show that it is possible to efficiently bound the fidelity of any such
experiment to a class of unitaries capable of universal quantum
computation.  Finally, in \sec{fid}, we analyze error propagation in
these protocols. This section includes new bounds on the effect of
imperfect randomizing operations, and rigorous bounds on the number of
samples needed to achieve some desired error and confidence.

\section{Completely Positive Trace Preserving Maps: Notation and Properties}
\label{sec:background}
Throughout this paper, we will restrict the discussion to
Hermiticity-preserving linear operations on quantum states---more
specifically, linear operations on multiqubit states, so that the
Hilbert space dimension will always be $d=2^n$ for $n$ qubits. The
physical operations within this class that are commonly considered are
completely-positive (CP) trace-preserving (TP) operations~\cite{Jam72,
  Cho75, Kra83}. We refer to these operations on quantum systems as
maps, and will denote them by the calligraphic fonts $\sop A, \sop B,$
etc. The composition of two maps will be denoted $\sop A\circ\sop B$,
meaning $\sop B$ acts on a state first, and then $\sop A$ acts.  Even
when discussing unitary evolution, we will refer to the corresponding
maps.  The notable exceptions are the identity unitary $\I$, and the
unitaries in the multi-qubit Pauli group $\Pauli$, which will be
denoted $\op P_i$---although the corresponding maps $\sop I$ and $\sop
P_i$ will also be used in some contexts. We will use the standard
convention where $\op P_0=\I$. We use $\sop T$ to mean the map
corresponding to the unitary $e^{-i\frac{\pi}{8}\op{Z}}$.

A map $\sop E$ is TP iff $\tr\op\rho=\tr\sop
E(\op\rho)$ for all $\op\rho$, which in turn leads to the requirement
that $\sop E^\dagger(\I)=\I$, where $\sop E^\dagger$ is the Heisenberg picture
representation of $\sop E$.  Any linear map $\sop E$ 
can be written as
\begin{align}
\sop E(\op\rho)=\sum_{i,j=0}^{d^2-1}\chi_{ij}^{\sop E}\op{P}_i\op\rho\op{P_j},
\end{align}
which is known as the {\em $\chi$ matrix representation of $\sop
  E$}. The map $\sop E$ is CP iff $\chi^{\sop E}$ is positive
semidefinite, and the TP condition $\sop E^\dagger(\I)=\I$ translates
to $\sum_{ij}\chi^{\sop E}_{ij}\op P_j\op P_i=\I$, which implies $\tr
\chi^{\sop E}=1$~\cite{CN97}. A map $\sop E$ is {\em unital} if
$\sop E(\I)=\I$.

It is often necessary to compute the representation of the composition
of two maps. While such a calculation can be cumbersome in the
$\chi$ representation, {\em Liouville representations} are more
convenient for describing the action of composed
maps on quantum states~\cite{Blum81}. In the Liouville representation, an
operator $\op \rho$ is represented by a column vector
$\kket{\op\rho}$, and maps are represented by matrices acting on
these vectors, such that the composition of maps corresponds to
matrix multiplication. The most convenient choice of basis for these
vectors and matrices depends on the application, but for our purposes
we will use the basis of Pauli operators, and will call this {\em the
  Pauli-Liouville representation} (which appears to have no standard
name in the literature, despite being widely
used~\cite{Leu00,KR01,Hav02,RDM02,SMKE08,CGC+12}). For a map $\sop E$, the
Pauli-Liouville representation is given by
\begin{align}
\sop E^\PL = \sum_{i,j=0}^{d^2-1} {\tr[\sop E(\op P_i)\op P_j]\over d} \ketbra{i}{j},
\end{align}
where $\op P_{i}$ and $\op P_{j}$ are $n$-qubit Pauli
operators. Hermiticity preservation implies that all matrix elements
of $\sop E^\PL$ are real. The $k^{\rm{th}}$ entry in the vector
$\kket{\op\rho}$ representing a density matrix $\op\rho$ corresponds
to $\tr \op\rho\op P_k$. This ensures that the Pauli-Liouville
representation of any CPTP map can be written as~\cite{Leu00,KR01}
\begin{align}
\sop E^{\PL}=\left(
\begin{array}{cc}
1& \vec{0}^T\\
\vec{\tau}_{\sop E} &\textbf{E}
\end{array}\right)
\end{align}
where $\vec{\tau}_{\sop E}$ is a $d^2-1$ dimensional column vector,
$\vec{0}$ is the corresponding zero vector, and $\textbf{E}$ is a
$(d^2-1)\times(d^2-1)$ matrix.

We will quantify how distinct a map $\sop E$ is from a particular
unitary map $\sop U$ by the average fidelity $\Favg(\sop E, \sop U)$,
which can be written as
\begin{align}
  \Favg(\sop E, \sop U)=\int d\mu(\psi)~\bra{\psi} (\sop U^\dagger \circ\sop
  E (\ketbra{\psi} {\psi}))\ket{\psi},
\end{align}
with integration taken over the unitarily invariant Fubini-Study
measure~\cite{BZ08}. This definition also implies $\Favg(\sop E, \sop
U) = \Favg(\sop E\circ \sop U^\dagger,\sop I) = \Favg(\sop U^\dagger
\circ \sop E,\sop I)$. The average fidelity is closely related to the
trace overlap between $\sop E^\PL$ and $\sop U^\PL$, as well as to
$\chi_{00}^{\sop E\circ\sop U^\dagger}$, by the
formulas~\cite{HHH99,nielsen02}
\begin{align}
\overline{F}(\sop E, \sop U) 
&= \dfrac{\tr \sop U^\dagger \sop E +d}{d(d+1)},\\
&= \dfrac{\chi_{00}^{\sop U^\dagger\circ\sop E}d+1}{d+1}.
\label{eq:Frelations-chi}
\end{align}
For simplicity and clarity, here, and throughout the paper, we omit
the superscripts from the Pauli-Liouville representation of
superoperators whenever they ocurr within trace expressions, as these
expressions already include superscripts indicating Hermitian
conjugates.

\section{Randomized Benchmarking of Clifford Group maps}

\label{sec:RB}

Randomized benchmarking (RB)~\cite{EAZ05,KLR+08,MGE11,MGE12,GMT+12,MGJ+12} consists
of a family of protocols to robustly estimate the average fidelity
$\overline{F}(\sop E, \sop U)$ between an experimental quantum map
$\sop E$ and an ideal unitary map $\sop U$. In this context, {\em
  robustness} refers to the ability to estimate $\overline{F}(\sop E,
\sop U)$ in a manner that is insensitive to imprecise or even biased
knowledge about state preparation, measurement, and controlled unitary
evolution. Such imperfections can lead to systematic errors, e.g., in
fidelity estimates based on standard tomographic reconstruction
protocols~\cite{MGS+12}.

We now describe a framework that can be used to understand existing RB
protocols, but which allows us to highlight how our protocol differs
from previous procedures.  RB protocols consist of $k$ repeated
applications of $\sop E$, each time preceded by independently chosen
randomizing unitary maps $\sop D_i$ where $1\le i\le k$, and, after the last application
of $\sop E$, followed by a recovery map $\sop D_{k+1}$. The
randomizing unitaries are chosen such that, if (i) the sequence is
applied to a fixed initial state $\ket{\psi}$, (ii) $\sop E$ is
identical to a certain unitary map $\sop U$, and (iii) the randomizing
maps $\sop D_i$ are perfect, then the final state would be
identical to the initial state.  If the first $k$ randomizing
operations are chosen from the Haar measure over unitary
maps~\cite{EAZ05,BZ08} or from a set with the same first- and
second-order moments as the Haar measure~\cite{DCEL09}, the fidelity
between the initial and final states can be shown to decay
exponentially with $k$ at a rate that depends only on
$\overline{F}(\sop E, \sop U)$~\cite{EAZ05,MGE11,MGE12}.  The RB
literature typically assumes either (1) $\sop U=\sop I$ and $\sop E$
represents the errors from the randomizing operations, or (2) $\sop U$
is some other unitary map, and $\sop E$ is its potentially
faulty implementation.  However, we emphasize our description is more
general, and as we will demonstrate later, allows us to reconstruct a
major portion of \textit{arbitrary} $\sop E$, not just implementations
of the randomizing operations.

In a realistic setting one cannot assume that the initial state is
pure and exactly known, that one knows what observable is measured
exactly, or that the randomizing operations are applied
noiselessly. However, these assumptions are not necessary for the RB
protocol to work: the initial state can be any mixed state $\op
\rho_0\not={1\over d}\I$, the measured observable $\op M$ can be any
observable where $\tr \op \rho_0\op M\not={1\over d}\tr \op M$, and
the rate of decay $p$ of the measured expectation value is still
related to $\overline{F}(\sop E, \sop U)$ in the same way. The
randomizing operations need not be noiseless
either~\cite{MGE11,MGE12}, as long as the imperfect randomizing
operations correspond to $\sop N\circ\sop D_i$, with $\sop N$
representing some arbitrary CPTP error map (some of these restrictions
may be relaxed, leading to more complex decays~\cite{MGE11,MGE12}, and
although our protocols generalize straightforwardly to such scenarios
we do not discuss them here for the sake of brevity). Under these more
realistic assumptions, $F_k(\sop E, \sop U) $, the average of $\langle
\op M\rangle$ over the choice of randomizing operations, for sequences
of length $k$, is given by
\begin{align}
\label{eq:model}
F_k(\sop E, \sop U) = A_0 p^k + B_0,
\end{align}
where $A_0$ and $B_0$ are constants that contain partial information
about the preparation and measurement (including imperfections), and
\begin{align}
\label{eq:prelation}
p
&={d~\overline{F}(\sop E\circ \sop N, \sop U) - 1 \over d - 1},\\
&={\tr \sop U^\dagger \sop E \sop N - 1\over d^2 - 1}.
\end{align}
By estimating $F_k(\sop E, \sop U)$ for different values of $k$, it is possible
to isolate $p$ (which contains the desired information about $\sop E$)
from $A_0$ and $B_0$ (which contain the undesired information about
preparation and measurement), creating a protocol that is largely free of
systematic errors caused by imprecise knowledge of state 
preparation and measurement~\footnote{In full
  generality, $p$ corresponds to an eigenvalue of the map
  resulting from randomizing $\sop E$ by conjugation with elements of
  either the full unitary group or the Clifford
  group~\cite{KR01,RDM02,ESM+07,SMKE08}. The eigenvalue interpretation
  can be used to more clearly see how independence of the estimate
  from the initial and final states comes about, and it can also be
  more naturally generalized to cases where the randomizing operations
  are elements of the Pauli group~\cite{SMKE08}. Randomization over
  more general operations can also be considered~\cite{BRS07}.}.

Case (1) discussed above is the original scenario considered in the RB
literature~\cite{EAZ05,KLR+08,MGE11,MGE12}, where $\sop U=\sop I$ and
$\sop E=\sop I$, so the observed decay leads to a direct estimate of
$\overline{F}(\sop N, \sop I)$, i.e., how well the randomization
operations are implemented. Case (2) discussed above is the extension
of RB to the extraction of information about $\overline{F}(\sop E,\sop
U)$, where $\sop E$ is one of the randomizing operations in the
experiment and $\sop U$ is its unitary idealization. This is a recent
development sometimes referred to as {\em interleaved
  RB}~\cite{GMT+12,MGJ+12}, but we do not make such a distinction in
this paper. The previously known result in this case is that
$\overline{F}(\sop E,\sop U)$ can be bounded by experimentally
estimating $\overline{F}(\sop E\circ \sop N,\sop U)$ and
$\overline{F}(\sop N,\sop I)$, and in Section~\ref{sec:bounds} we
provide more general bounds (with fewer assumptions) for the same
purpose.

While the RB protocol is valid for any choice of randomizing
operations discussed above, we emphasize that, in order to ensure the
protocols remain scalable in the number of qubits, $\sop U$ and $\sop
D_i$ are restricted to be unitary maps in the Clifford group, since
this allows for scalable design of the randomizing sequences via the
Gottesman-Knill theorem~\cite{Got99}. Moreover, although previous
works have applied the RB protocols only to $\sop E$ very close to
Clifford group maps, we emphasize that no restriction beyond $\sop E$
being CPTP needs to be imposed. The restricted applications of the RB
protocols in previous work was partially due to the bounds used to
isolate $\overline{F}(\sop E,\sop U)$ being only useful when $\sop E$ is
close to a Clifford group map. Since we are interested in extracting
information about arbitrary $\sop E$, we consider here
tomographic reconstruction techniques that {\em do not} rely on these
bounds. We also design efficient techniques for average-fidelity estimates
that rely on new and improved general bounds on $\overline{F}(\sop E,\sop
U)$.

In summary, RB allows for efficient estimation of
$\overline{F}(\sop E\circ\sop N, \sop U)$ and
efficient bounding of $\overline{F}(\sop E, \sop U)$ for
    $\sop U$ in the Clifford group. These estimates can be obtained
without relying on perfect information about preparation and
measurement errors, thereby avoiding some of the systematic errors
that may be present in standard tomographic protocols due to these
imperfections.

\subsection{RB sequence design}

A compact way to describe how RB sequences are constructed refers back
to the idea of {\em
  twirling}~\cite{BDS+96,DCEL09,ESM+07,BRS07,MSR+12}. Although this is
not how this construction is typically described, we found it to be
convenient, and include it for completeness.

If $\sop E$ is an abritrary quantum map, and $\matheu{S}$ is a set of
maps $\{\sop C_0, \cdots \}$, the average map
\begin{align}
\mathbb{E}_i[\sop C_i^\dagger \circ\sop U^\dagger \circ\sop E \circ \sop C_i] =
{1\over |\matheu{S}|}\sum_{\sop C_i\in \matheu{S}} \sop C_i^\dagger \circ \sop U^\dagger \circ\sop E \circ\sop C_i,
\end{align}
is called the twirl of $\sop U^\dagger \circ \sop E$ over
$\matheu{S}$, where $\mathbb{E}_i$ denotes the expectation value over
uniformly random choices for $\sop C_{i}\in\matheu{S}$. If
$\matheu{S}$ is the Clifford group or any other unitary
2-design~\cite{DCEL09}, then
\begin{align}
\mathbb{E}_i[\sop C_i^\dagger \circ\sop U^\dagger \circ\sop E \circ \sop C_i (\op \rho)] = p \op \rho + {(1-p)\over d} \I
\end{align}
where $p = {\tr \sop U^\dagger \sop E - 1 \over d^2 - 1} = {d \overline{F}(\sop E,
  \sop U) - 1\over d - 1}$ as before. 

A length $k$ RB sequence consists of applying the twirled channel
repeatedly to the same state $k$ times, i.e.,
\begin{align}
\label{eq:rb-seq}
& \mathbb{E}_{\vec{i}}[ 
\sop C_{i_k}^\dagger \circ\sop U^\dagger \circ\sop E \circ\sop C_{i_k} \circ\cdots\circ 
\sop C_{i_1}^\dagger \circ\sop U^\dagger \circ\sop E \circ\sop C_{i_1} (\op \rho) ]\notag,\\
=& \mathbb{E}_{\vec{i}}[ 
\sop D_{i_{k+1}} \circ\sop E \circ \sop D_{i_{k}} \circ\sop E \circ \sop D_{i_{k-1}}\circ\cdots\circ 
\sop D_{i_{2}} \circ\sop E \circ\sop D_{i_1} (\op \rho)]\notag,\\
=& p^k \op \rho + {(1-p^k)\over d} \I,
\end{align}
where 
\begin{align}
\sop D_{i_\ell} = \left\{ 
\begin{array}{ll}
\sop C_{i_{1}},& \ell=1,\\
\sop C_{i_{\ell}} \circ \sop C_{i_{\ell-1}}^\dagger \circ \sop U^\dagger , & 1<\ell<k\\
\sop C_{i_{k}}^\dagger \circ \sop U^\dagger,& \ell=k+1,
\end{array}
\right.
\end{align}
and $\mathbb{E}_{\vec{i}}$ denotes the expectation value over uniformly random choices 
for $\sop C_{i_\ell}\in\matheu{S}$ for all $\ell$.

The RB protocol to estimate $\overline{F}(\sop E, \sop U)$ then
consists of (i) choosing sequence of $\sop C_{i_\ell}$ for $1<\ell\le
k$, (ii) applying the alternating sequence of $\sop D_{i_\ell}$ and
$\sop E$, as prescribed in \eqref{eq:rb-seq}, to a fixed initial
state, (iii) measuring the resulting state, and (iv) averaging over
random choices for $\sop C_{i_\ell}$ to obtain $F_k$. The $F_k$ can be
fit against \eqref{eq:model}, yielding an estimate for $p$, even in
the presence of imperfections. As we prove in Section~\ref{sec:fid},
this estimate can be obtained efficiently in the number of qubits,
desired accuracy, and confidence. Note that neither $\sop E$ nor $\sop
U$ need to be elements of the Clifford group. However, we will
generally consider the case where $\sop E$ is \textit{not} a Clifford
group map, while $\sop U$ \textit{will} be chosen to be a Clifford
group map. Choosing $\sop U$ to be a Clifford group element makes the
design of the experiments for $n$-qubits
efficient~\cite{Got99,MGE12}, while leaving $\sop E$ unconstrained
affords us greater flexibility and has no impact on experiment design.

\section{Tomographic Reconstruction from RB}
\label{sec:tomography}

As discussed above, RB can efficiently provide bounds
on the fidelities of an arbitrary CPTP map $\sop E$ with any
element of the Clifford group---in a manner that is robust against
preparation and measurement errors, as well as imperfections in the
twirling operations. Here we demonstrate that the collection of such
fidelities of a fixed $\sop E$ to a set of linearly independent
Clifford group maps can be used to reconstruct a large portion of
$\sop E$. The advantage of this approach is that the robustness
properties of the estimates obtained via RB carry over to this 
tomographic reconstruction.

Using the Liouville representation of quantum maps, it is clear that an
estimate of the average fidelity $\overline{F}(\sop E, \sop U)$ leads
to an estimate of $\tr\sop U^\dagger \sop E$, and thus all information
that can be extracted from these fidelities for a fixed $\sop E$ is
contained in the projection of $\sop E$ onto the linear span of
unitary maps.  It turns out to be unnecessary to consider the span of
arbitrary unitary maps, as the following result demonstrates (see
Appendix~\ref{app:clifford-span} for the proof).

\begin{restatable}{lem}{cliffordspan} 
\label{lemm:cliff-span-unital}
  The linear span of unitary maps coincides with the linear span of
  Clifford group unitary maps. Moreover, the projection of a TP map to
  this linear span is a unital map.
\end{restatable}

Given a set of linearly independent vectors that span a subspace, and
the inner product of an unknown vector with all elements of that set,
it is a simple linear algebra exercise to determine the projection of
the unknown vector onto the subspace. Similarly, measuring the average
fidelity of some TP map $\sop E$ to a Clifford group map $\sop C_i$ is
equivalent to measuring such an inner product---the matrix inner
product $\tr(\sop E\sop C_i^\dagger)$.  Since Clifford maps span the
unital subspace of quantum CPTP maps, measuring the inner product of
$\sop E$ with a set of maximal linearly independent elements of the
Clifford group is sufficient to reconstruct the projection of $\sop E$
onto the unital subspace. We call this projection the {\em unital
  part} of $\sop E$, and denote it by $\sop E'$.

Since the unitality condition constrains
only how the map acts on the identity component of a state, $\sop E'$
can be obtained by changing how $\sop E$ acts on that
component. Defining $\sop Q$ to be the projector into the identity
component of any operator, and $\sop Q^\perp$ to be the projection
into the orthogonal complement (i.e. $\sop Q + \sop Q^\perp = \sop
I$), one finds that
\begin{align}
  \sop E &= \sop E \circ (\sop Q^\perp + \sop Q) 
= \sop E \circ \sop Q^\perp + \sop E \circ \sop Q,\\
  \sop E' & = \sop E \circ \sop Q^\perp + \sop Q,
\end{align}
which indicates that $\sop E$ and $\sop E'$ map traceless operators in
the same way. The maps $\sop E$ and $\sop E'$ have Pauli-Liouville representations
\begin{align}
\label{eq:nonunitalpart}
\sop E^{\PL}& =\left(
\begin{array}{cc}
1& \vec{0}^T\\
\vec{\tau}_{\sop E} &\textbf{E}
\end{array}\right),
&
\sop {E'}^{\PL}& =\left(
\begin{array}{cc}
1& \vec{0}^T\\
\vec{0}&\textbf{E}
\end{array}\right),
\end{align}
so we refer to $\vec\tau_{\sop E}$ as the non-unital part of $\sop E$.
It is then clear that $\sop E'$ is described by $(d^{2}-1)^2$ real
parameters if $\sop E$ is TP, while $\sop E$ itself is described by
$(d^{2}-1)d^{2}$ real parameters. The unital part of $\sop E$ contains
the vast majority of the parameters needed to describe $\sop E$---in
fact, over $93\%$ of the parameters for two qubits, and over $99\%$
of the parameters for four qubits.

As discussed, one limitation of RB is that in a realistic setting it
can only provide bounds for $\overline{F}(\sop E, \sop C_i)$ (and
therefore $\tr \sop E\sop C_i^\dagger$) due to the imperfections in
the randomizing operations. Clearly these bounds can only lead to a
description of parameter-space regions compatible with $\sop E'$ as
opposed to any point estimator, even in the absence of statistical
fluctuations. Our approach to reconstruct $\sop E'$ is to avoid these
bounds altogether and instead use the following result, which we prove
in Appendix~\ref{app:unital-recon}.

\begin{restatable}{lem}{noisyunitalest} 
\label{lemm:noisy-unital-est}
  If $(\sop E\circ\sop N)'$ is the unital part of $\sop E\circ\sop N$
  and $\sop N'$ is the unital part of $\sop N$, and all these
  operations are trace preserving, then $\sop E' =
  (\sop E\circ\sop N)'\circ(\sop N')^{-1}$ whenever $(\sop N')^{-1}$
  exists.
\end{restatable}

This allows us to reconstruct $\sop E'$ from the reconstructions of
$(\sop E\circ \sop N)'$ and $\sop N'$. As both $(\sop E\circ \sop N)'$
and $\sop N'$ are related directly to decay rates, we can create a
point estimate of $\sop E'$, without recourse to the bounds needed
in standard RB to characterize $\sop E$.

It should be noted that the only cases where $(\sop N')^{-1}$ does not
exist are when $\sop N$ completely dephases some set of observables
(i.e., maps them to something proportional to the identity). However,
the experimental setting where tomographic reconstructions are
interesting are precisely in the regime where $\sop N$ is far from
depolarizing any observable, so that $(\sop N')^{-1}$ is typically
well defined~\footnote{For $\sop N$ chosen at random to have
  unitary dilations that are Haar distributed, $(\sop N')^{-1}$
  appears to exist with probability 1, so it appears the requirement
  that $\sop N$ be close to $\sop I$ can be significantly
  weakened.}. The penalty, of course, is that the application of
$(\sop N')^{-1}$ leads to greater statistical uncertainty in the
estimate of $\sop E'$ thanks to the uncertainties in the
reconstructions of $\sop N'$ and $(\sop E\circ\sop N)'$ as well as
uncertainty propagation due to multiplication by $(\sop N')^{-1}$, but
larger experimental ensembles can be used to compensate for this, as
is discussed in the section that follows.

Moreover, writing the imperfect randomizing operations as $\sop N
\circ \sop C_i$ instead of $\sop C_i \circ \sop N^*$ for some
different map $\sop N^*$ is merely a convention, and
Lemma~\ref{lemm:noisy-unital-est} can be trivially adjusted to such a
different convention. In the physical regimes where RB estimates are
expected to be valid, the choice of conventions is largely immaterial
(see Appendix~\ref{app:gauge} for more details).

This result shows that the average fidelities with a spanning set of
Clifford group unitary maps can lead, not only to a point estimator of
the unital part of any TP map, but also to a point estimator of the
average fidelity of $\sop E$ to any unitary map---i.e., information
from multiple RB experiments can eliminate the need for the loose
bounds on the average fidelity considered in~\cite{MGJ+12}. This comes
at the cost of efficiency, as the unital part of a map---like the
complete map---contains an exponential number of parameters. However,
for a small number of qubits the overhead of reconstructing the unital
part is small, and therefore it is still advantageous to perform this
cancelation to get better estimates of the error.

\subsection{Example: Single Qubit Maps}

In order to reconstruct the unital part of a single-qubit map, one must
first consider a set of linearly-independent maps corresponding to
unitaries in the Clifford group. As this group contains 24 elements,
there are many different choices for a linearly independent set
spanning the 10-dimensional unital subspace. One particular choice of
unitaries leading to linearly independent maps is
\begin{align}
\op C_0&=\I, & 
\op C_1&=\e^{-i{\pi\over2}\op X}, \\
\op C_2&=\e^{-i{\pi\over2}\op Y}, & 
\op C_3&=\e^{-i{\pi\over2}\op Z},\\
\op C_4&=\e^{-i{\pi\over3}{\op X + \op Y + \op Z\over \sqrt{3} }}, & 
\op C_5&=\e^{-i{2\pi\over3}{\op X + \op Y + \op Z\over \sqrt{3} }}, \\
\op C_6&=\e^{-i{\pi\over3}{\op X - \op Y + \op Z\over \sqrt{3} }},  & 
\op C_7&=\e^{-i{2\pi\over3}{\op X - \op Y + \op Z\over \sqrt{3} }},\\ 
\op C_8&=\e^{-i{\pi\over3}{\op X + \op Y - \op Z\over \sqrt{3} }},  & 
\op C_9&=\e^{-i{2\pi\over3}{\op X + \op Y - \op Z\over \sqrt{3} }}.
\end{align}
In a noiseless setting, estimating the average fidelities between
these Clifford maps and the map
\begin{align}
\sop H^\PL & = 
\left(
\begin{array}{rrcr}
1 & \phantom{-}0 & \phantom{-}0 & \phantom{-}0 \\
0 & \phantom{-}0 & \phantom{-}0 & \phantom{-}1 \\
0 & \phantom{-}0 & -1 & \phantom{-}0\\
0 & \phantom{-}1 & \phantom{-}0 & \phantom{-}0
\end{array}
\right),
\end{align}
corresponding to the single-qubit Hadamard gate, leads to the decays
illustrated in Figure~\ref{fig:decays}.  The corresponding $p$ values are
\begin{align}
p_{0}=p_{2}=p_{8}=p_{9}&=-{1\over3},\\
p_{1}=p_{3}=p_{4}=p_{5}=p_{6}=p_{7}&=\phantom{-}{1\over3}.
\end{align}
Note, in particular, that some $p$ values are negative, which simply
indicates an oscillatory exponential-decay behaviour. While these
decay rates are much larger (i.e., the $p$ values are much smaller) than
those typically seen in previous RB protocols, we show in
\sec{conf-bounds} that it is possible to efficiently estimate any
decay rate to fixed accuracy, no matter the size.

\begin{figure}
\includegraphics[width=.45\textwidth]{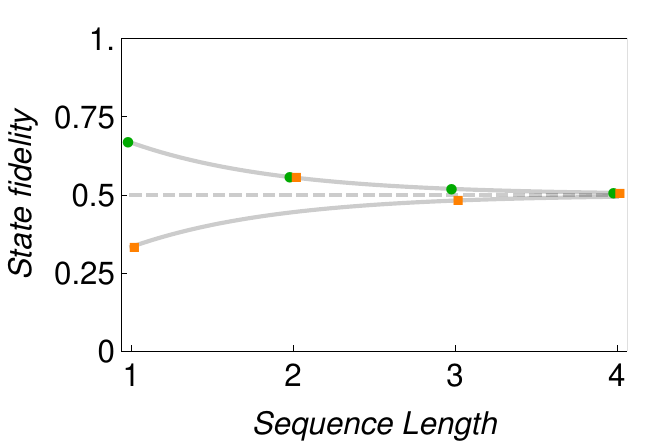}
\caption{RB decays used to estimate the fidelity between an ideal
  Hadamard gate and $\op C_0$ (green circles, with $p={1\over3}$), and
  $\op C_1$ (orange squares, with $p=-{1\over3}$). The decays corresponding to each
  of the remaining average fidelities coincide with one of these two representative
  decays.  Note that these decays are much faster than decays
  previously estimated in RB, as they corresponds to the average fidelities
  between very different maps.  The data points are offset along the x-axis
  for clarity.\label{fig:decays}}
\end{figure}

If one considers a noisy setting, where $\sop N$ is not the identity,
the decay rates are modified by $\sop N$, but after reconstructing
$\sop N'$ and $(\sop E \circ \sop N)'$ separately, one is able to
reconstruct $\sop E'$. To see that errors in the estimate of $\sop N'$ will not create
unmanageable errors in the estimate of $\sop E$, consider how errors in the estimate
of $\sop N'$ affect the estimate of $(\sop N')^{-1}$. The relative
error in the estimate of $(\sop N')^{-1}$ is given by~\cite{HJ85}
\begin{align}
{\|(\sop N')^{-1}-(\sop N' + \sop G')^{-1} \| 
\over 
\| (\sop N')^{-1} \|}
\le
{\kappa(\sop N')\over 1- \kappa(\sop N') 
{\|\sop{G'}\| \over \|{\sop N'}\|}} {\|\sop{G'}\|\over \|{\sop N'}\|},
\end{align}
as long as
\begin{align}
\|\sop{G'}\|\|(\sop N')^{-1}\|<1,
\end{align}
where $\sop G'$ is the error in the estimate of $\sop N'$, and
$\kappa(\sop N')$ is the condition number for the matrix inversion of
$\sop N'$ with respect to the matrix norm ${\|\cdot\|}$. The condition
number of $\sop A$ is given by $\kappa(\sop A) = \|\sop A^{-1}\| \|\sop A\|$ if
$\sop A$ is non-singular, and $\infty$ otherwise.

If we choose $\|\cdot\|$ to be the spectral norm, even when $\sop N'$
is the depolarizing map $\sop
D(\op\rho)=\delta\op\rho+(1-\delta){\I\over d}$, the condition number
of $\sop N'$ is given by $\kappa(\sop
N')={1\over|\delta|}$. Similarly, if $\sop N'$ is the dephasing map
$\sop Z(\op\rho)={1+\gamma\over2}\op\rho+{1-\gamma\over 2}\op
Z\op\rho\op Z$, one finds $\kappa(\sop N')={1\over|\gamma|}$. Thus,
even for $\delta$ and $\gamma$ polynomially close to 0, a polynomial increase in the 
number of statistical samples can be used to ensure an estimate of the
inverse of $\sop N$ to any polynomial accuracy with high probability.

\subsection{Beyond Unital Maps}
\label{sec:non-unital}

What does the reconstruction of $\sop E'$ tell us about
the $\sop E$? We prove in \app{projection} that 
\begin{restatable}{lem}{unitalproject} \label{lemm:project}
The unital part of a CPTP single-qubit 
map is always a CPTP map.
\end{restatable}
This means that the unital part of a single-qubit map
imposes no lower bound on the magnitude of the non-unital part of that
map---the non-unital part can always be set to 0.

For a single qubit, the unital part does impose stringent conditions
on the maximum size of the non-unital part. Up to unitary rotations,
any map can be written in the Pauli-Liouville representation as
~\cite{KR01}
\begin{align}\left(
\begin{array}{cccc}
1&0&0&0\\
t_1& \lambda_1&0&0\\
t_2&0& \lambda_2&0\\
t_3&0&0 &\lambda_3\\
\end{array}\right),
\end{align}
where $\lambda_i$ and $t_i$ are real valued parameters.  The
$\lambda_i$, corresponding to the unital part, can be estimated using
the techniques already described, but as \lemm{project} demonstrates,
no useful lower bound on $|t_i|$ can be obtained.  However, for the
map to be positive, it is necessary that $|t_i|\leq
1-|\lambda_i|$~\cite{KR01}, which gives upper bounds on the magnitudes
of the non-unital parameters.

The fact that, for single-qubit maps, $\sop E'$ is always CP can be
turned around to say that statistically significant non-CP estimates
of $\sop E'$ imply statistically significant non-CP estimates of $\sop
E$, and may be used as witnesses of systematic errors in the
experiments~\cite{WHE+04,MKS+12}.

\lemm{project} fails in the case of multiple qubits, and it is not
difficult to construct counter-examples. Numerical experiments
indicate that CPTP maps chosen at random by drawing unitary dilations
from the Haar distribution lead to non-CP unital parts with
probability $\sim 1$. This implies that, while it may not be possible
to test complete-positivity of a general map by testing only its
unital part, the reconstruction of the unital part of a multi-qubit
map yields lower-bounds on the magnitudes of the non-unital
parameters. Thus, while this result precludes the use of the unital
part of a multi-qubit map to test for systematic errors in
experiments, it does provide more information about the non-unital
parameters.

\section{Fidelity Estimation Beyond the Clifford Group}\label{sec:fidTandU}

Previous RB results showed how to bound the average fidelity of
Clifford operations~\cite{MGE12,MGJ+12}.  While the maps in the
Clifford group form an integral part of current approaches to scalable
fault-tolerance in quantum computers, universal quantum computation is
only possible if operations outside the Clifford group are also
considered. We would like to be able to not only efficiently verify
the performance of Clifford gates, but also would like to be able to
verify the performance of universal circuits. However, there are
strong indications that quantum computers are strictly more powerful
than classical computers; for example, if classical computers could
efficiently simulate certain classes of non-universal quantum
circuits, it would imply a collapse of the polynomial
hierarchy~\cite{Bremner2011}, and so is considered highly unlikely. It
is therefore extremely unlikely that classical computers can
efficiently predict the behaviour of a general
${\mathrm{poly}}(n)$-depth quantum circuit~\footnote{Here we take a
  {\em circuit} to mean a composition of quantum maps on $n$ qubits,
  and the depth to correspond to the number of maps composed.}, and
without these predictions, it is not possible to check if a quantum
computer is behaving as desired. For these fundamental reasons, we do
not expect that it is possible to efficiently estimate the average
fidelity to a general quantum map.

It is important to note, however, that it is possible to efficiently
simulate {\em some} circuits that contain maps outside the Clifford
group.  In particular, Aaronson and Gottesman~\cite{Aaronson2004} have
proven that circuits consisting of Clifford group maps and a
logarithmic number of maps outside the Clifford group can be simulated
efficiently. Despite being efficiently simulatable, these circuits can
be though of as discrete components that enable universality under
composition, and thus the ability to verify their implementation is of
great practical importance. We now show how our methods can be
extended to allow for efficient estimation of the average fidelity of
{\em any} experiment to such circuits.

\subsection{Average Fidelity to $\sop T$}
\label{sec:fid-T}

The canonical example of a map outside the Clifford group is the
operation $\sop T=e^{-i\frac{\pi}{8}\op{Z}}$. This gate is commonly
used in combination with Clifford group operations to form a gate set
that is universal for quantum computation~\cite{BMP+99}.  In this
section we show how to efficiently bound the average fidelity of a map
$\sop E$ to $\sop U=\sop T$.

In \sec{tomography} we prove that Clifford maps span the space of
unital maps. This implies that, in the Pauli-Liouville representation,
any unitary map $\sop U^\PL$ can be written as a linear combination of Clifford maps
\begin{align}
\label{eq:lincomb}
\sop U^\PL=\sum_i\beta_i^{\sop U}\sop C_i^\PL,
\end{align}
with $\beta_i\in\mathbb{R}$.
By linearity,
\begin{align}
\tr\sop E\lamCn\sop U^\dagger=\sum_i\beta_i^{\sop U}\tr\sop E\lamCn\left(\sop
C_i\right)^\dagger,
\end{align}
so
\begin{align}
\label{eq:fid-sum}
\Favg(\sop E\circ\lamCn,\sop U)=\sum_i\beta_i^{\sop U}\Favg(\sop E\circ\lamCn,\sop
C_i)+\dfrac{1}{d+1}\left(1-\sum_i\beta_i^{\sop U}\right).
\end{align}

For an arbitrary unitary $\sop U$, the number of non-zero $\beta_i^{\sop U}$, which we
denote by $N_\sop U$, can be as large as $O(d^2)$. However $\sop
T^\PL$ can be written as a linear combination of three Clifford maps. The
support of $\sop T^\PL$ is given by the maps corresponding to the Clifford
group unitaries, $\I$, $\op Z$, and $e^{-i\frac{\pi} {4}\op Z}$, with
the corresponding coefficients $\frac{1}{2}$, $\frac{1-\sqrt{2}} {2}$,
and $\frac{1}{\sqrt{2}}$.  Thus, to estimate $\Favg(\sop
E\circ\lamCn,\sop T$), one only needs to estimate 3 average fidelities
to Clifford group maps (instead of the 10 necessary for reconstruction
of the unital part).

Suppose one estimates each fidelity $\Favg(\sop E\circ\lamCn,\sop
C_i)$ for all of the $\sop C_i$ in the linear combination to within $\epsilon'$ with confidence $1-\delta'$. 
In \sec{conf-bounds} we will show this requires  
$O\left ({N_\sop U\over{\epsilon'}^4}\log{1\over\delta'}\right)$
samples. From \eq{fid-sum} it is clear that one can obtain an  
estimate $\widetilde{F}$ such that
\begin{align}
\Pr\left(|\widetilde{F}-\Favg(\sop E\circ\lamCn,\sop U)|\ge\epsilon'\sum_i|\beta_i^{\sop U}|\right)\le N_\sop U\delta'.
\end{align}
Choosing $\delta'=\delta/N_\sop U$ and $\epsilon'=\epsilon/\sum_i|\beta_i^{\sop U}|$ gives
\begin{align}
\Pr(|\widetilde{F}-\Favg(\sop E\circ\lamCn,\sop U)|\ge\epsilon)\le\delta,
\end{align}
and requires $O\left (N_\sop U{\left(\sum_i|\beta_i^{\sop
        U}|\over{\epsilon}\right)^4}\log{N_{\sop
      U}\over\delta'}\right)$ samples.

For the particular case of the $\sop T$ map, one finds
$\sum_i|\beta_i^{\sop T}|=\sqrt{2}$, so an estimate for the average
fidelity to $\sop T$ can be obtained by the following procedure:
\begin{enumerate}
\item Perform RB with $O\left({1\over{\epsilon}^4}\log{1\over
      \delta}\right)$ samples for each relevant fidelity
  $\Favg(\sop E\circ\lamCn,\sop C_i)$ . This requires
  $O\left({1\over{\epsilon}^4}\log{1\over \delta}\right)$ total
  samples and results in an estimate $\widetilde{F}$ such that
\begin{align}
\label{eq:conf}
P(|\widetilde{F}-\Favg(\sop E\circ\lamCn,\sop T)|
\ge\epsilon)\le\delta.
\end{align}
\item Perform RB with $O\left({1\over{\epsilon}^4}\log{1\over
      \delta}\right)$ samples to obtain an estimate $\widetilde{F}_{\lamCn}$ of
  $\Favg(\lamCn,\sop I)$ such that
\begin{align}
P(|\widetilde{F}_{\lamCn}-\Favg(\lamCn,\sop I)|\ge\epsilon)\le\delta, 
\end{align}
\item In \sec{bounds}, we show how to bound the fidelity of $\Favg(\sop E,\sop U)$,
given estimates of $\Favg(\sop E\circ \sop N,\sop U)$ and $\Favg(\sop N,\sop I)$. Apply
the bounds of \sec{bounds} for $\Favg(\sop E\circ\sop N,\sop T)=\widetilde{F}\pm\epsilon$,
and for $\Favg(\sop N,\sop I)=\widetilde{F}_\sop N\pm \epsilon$, to obtain bounds
on $\Favg(\sop E,\sop T)$ that are valid with probability
  at least $1-2\delta$.
\end{enumerate}

This procedure trivially extends to bounding the fidelity of $\sop E$
to the case where $\sop T$ acts on a single qubit and the identity
acts on $n-1$ qubits. The sampling complexity remains the
same, but the time complexity changes, as the classical preprocessing
time needed to make a single average fidelity estimate scales as
$O(n^4)$~\cite{MGJ+12}.  Similar arguments can be used to show that
the sampling complexity of determining the average fidelity of $\sop
E$ to any 1- or 2-qubit unitary acting on $n$ qubits is constant, with
the same classical preprocessing time complexity. In the next section,
we will discuss more general operations acting on $n$ qubits.

\subsection{Average Fidelity to More General Unitaries}

It is possible to efficiently bound the average fidelity of a map
$\sop E$ to a unitary $\sop U$ when $\sop U$ is a
composition of $O(\text{poly}(n))$ Clifford maps and $O(\log(n))$
$\sop T$ maps on $n$ qubits (i.e., maps that acts as $\sop T$ on one
qubit and as the identity on the remaining $n-1$ qubits). Under these
constraints,
\begin{enumerate}[(i)]
\item $\sop U^\PL$ can be efficiently decomposed into a linear combination of 
$O(\textrm{poly}(n))$ Clifford maps. (i.e. $N_\sop U=O(\textrm{poly}(n))$)
\item The coefficients $\beta_i^{\sop U}$ in the linear combination satisfy 
$\sum_i|\beta_i^{\sop U}|=O(\textrm{poly}(n))$
\end{enumerate}
Following the argument of \sec{fid-T}, the sampling complexity scales like
$O\left (N_\sop U{\left(\sum_i|\beta_i^{\sop U}|\over{\epsilon}\right)^4}\log{N_{\sop U}\over\delta'}\right)$,
so together (i) and (ii) guarantee that the 
sampling complexity of bounding $F(\sop E,\sop U)$ is $O(\textrm{poly}(n))$.
Since (i) guarantees that the decomposition is efficient, and the classical
preprocessing time needed to make a single sample scales as $O(n^4)$,
the time complexity is also $O(\textrm{poly}(n))$.

We prove (i) by induction on $t$, the number of $\sop T$ maps in the 
circuit, and $c$, the number of Clifford maps in the circuit. We show one can
decompose $\sop U^\PL$ into a linear combination of at most $3^t$ terms, where
each Clifford map in the linear combination is written as a composition of at most
$t+c$ Clifford maps.
The base case is given by
\begin{itemize}
\item $t=1$, $c=0$: $\sop U$ is a single $\sop T$, and $\sop U^\PL$ can
be written as a linear combination of $3$ Clifford maps.
\item $t=0$, $c=1$: $\sop U$ is a Clifford and so $\sop U^\PL$ can
be written as a linear combination of $1$ Clifford map.
\end{itemize}

For the inductive case, assume one has a unitary $\sop U$ which is a 
composition of $t$ $\sop T$ maps and $c$ Clifford maps.
By inductive assumption, $\sop U^\PL$ can be written as
\begin{align}
\sop U^\PL=\sum_{i=1}^{M}\beta_i^{\sop U}\prod_{j=1}^{N_i}\sop C_{i,j}^\PL,
\end{align}
with $M\leq 3^t$ and $N_i\leq t+c$. Now consider
composing $\sop U$ with a Clifford $\sop C$. Then 
\begin{align}
\sop C^\PL\sop U^\PL=\left(
\sum_{i=1}^{M}\beta_i^{\sop U}\sop C^\PL\prod_{j=1}^{N_i}\sop 
C_{i,j}^\PL\right) ,
\end{align}
and one obtains a linear combination of $\leq 3^t$ terms, each a
composition of $c+t+1$ Clifford maps. Likewise, if $\sop U$ is
composed with $\sop T$, then
\begin{align}
\sop T^\PL\sop U^\PL =
\sum_{i=1}^{M}\sum_{k=1}^3\beta_i^{\sop U}\beta_k^{\sop T}
\sop C_k^{\sop T\PL}\prod_{j=1}^{N_i}\sop C_{i,j}^\PL,
\end{align}
where the $\sop C_k^\sop T$ are the three Clifford maps involved in the linear combination
of $\sop T^\PL$, so one obtains a linear combination of $3^{t+1}$ elements, each a composition
of $t+c+1$ Clifford maps, as desired.

Therefore, if one has a unitary map $\sop U$ composed of
$O(\text{poly}(n))$ Clifford maps and $O(\log(n))$ $\sop T$ maps, one
can write $\sop U^\PL$ as a linear combination of $O(\text{poly}(n))$ Clifford
maps, where each term in the linear combination is a composition of at
most $O(\text{poly}(n))$ Clifford maps. A sequence of
$O(\text{poly}(n))$ Clifford maps can be efficiently simplified into a
single Clifford map using the Gottesman-Knill Theorem~\cite{Got99}.
The average fidelity estimate to $\sop U$ is obtained by estimating
the average fidelities to these simplified Clifford maps.

To see that (ii) also holds, suppose one calculates a
linear combination for $\sop U^\PL$ based on the above construction. It is
possible that different terms in the linear combination result in the
same Clifford map, 
but for simplicity we treat each term separately,
so that our estimate of the complexity is an upper bound. Then if the
circuit decomposition of $\sop U$ contains $t$ $\sop T$ maps,
\begin{align}
\sum_i|\beta_i^{\sop U}|\le
\left(\sum_i\left|\beta_i^{\sop T}\right|\right)^{t}=\sqrt{2}^t,
\end{align}
so $\sum_i|\beta_i^{\sop U}|$ scales, at most, as
$O(\textrm{poly}(n))$ for $t=O(\log n)$.

These results demonstrate that robust estimates of the average fidelities to
unitary maps outside the Clifford group can be obtained efficiently,
scaling polynomially in the number of qubits.

\section{Bounding Error in Average Fidelity Estimates}
\label{sec:fid}

In this section, we bound sources of error that occur in RB procedures. 
There are two sources of uncertainty we consider. When trying to efficiently estimate
the average fidelity $\sop E$ without inverting $\sop N$, as we do in \sec{fidTandU}, we lack
 direct access to $\Favg(\sop E, \sop U)$ and instead can
only estimate $\Favg(\sop E\circ\sop N, \sop U)$ and $\Favg(\sop N, \sop
I)$. This leads to error on our estimate of $\Favg(\sop E, \sop U)$.
We also consider statistical error from the sampling of random
variables, and show that we can efficiently fit RB decays to any constant error. 
As a consequence, this allows us to efficiently bound the average fidelity to
maps outside the Clifford group, as described in \sec{fidTandU}. We address these two effects separately. These types
of uncertainties can be found in many contexts, so we expect the
analysis in \sec{bounds} and \sec{conf-bounds} has broader
applications.

\subsection{Bounds on Average Fidelity of Composed Maps}\label{sec:bounds}
In this section, we show how to bound $\Favg(\sop E,\sop U)$,
when you have estimates of $\Favg(\sop E\circ \sop N,\sop U)$ and $\Favg(\sop N,\sop I)$.

In \app{bound} we prove
\begin{align}
\label{eq:bound}
\chi_{0,0}^{\sop A\circ \sop B}=&\chi_{0,0}^{\sop A}\chi_{0,0}^{\sop 
B}\pm\nonumber\\
&\Big(2\sqrt{(1-\chi_{0,0}^{\sop A})\chi_{0,0}^{\sop A}
(1-\chi_{0,0}^{\sop B})\chi_{0,0}^{\sop B}}\nonumber\\
&+(1-\chi_{0,0}^{\sop A})(1-\chi_{0,0}^{\sop B})\Big).
\end{align}
Setting $\sop B=\sop N$, $\sop A=\sop U^\dagger\circ \sop E$, and using
 \eq{Frelations-chi} gives bounds on $\Favg(\sop E, \sop U)$ as a 
function of $\Favg(\sop E\circ\lamCn, \sop U)$ and $\Favg(\lamCn,\sop I)$.

This bound is valid for any maps $\sop A$ and $\sop B$.  There exist
$\sop A$ and $\sop B$ that saturate the upper bound, but the lower
bound is not tight, for reasons we discuss in \app{bound}. Generally,
this method gives better bounds when the operation $\sop E$ is close
to $\sop U$ and when $\sop N$ is close to $\sop I$ (i.e. the
imperfections in the randomizing operations are small). Because these
lower and upper bounds---just as the bounds in Ref.~\cite{MGJ+12}---are
not close to each other except in the regime where $\sop E$ is close
to $\sop U$, they are not useful for the type of
tomographic reconstruction performed in section \sec{tomography},
where an arbitrary map might be far from a Clifford map or from a map
that is composed of Clifford maps and $O(\textrm{poly}(n))$ $\sop T$ maps.

Previous work on average-fidelity estimates based on RB have derived the 
bound~\cite{MGJ+12}
\begin{align}
\label{eq:theirbound}
\chi_{0,0}^{\sop A}&=\dfrac{(d^2-1)\chi_{0,0}^{\sop A\circ\sop B}}
{d^2\chi_{0,0}^{\sop B}}\pm E\\
\label{eq:easwarbound2}
E&=
\left|\chi_{0,0}^{\sop B}-\dfrac{(d^2-1)\chi_{0,0}^{\sop A\circ\sop B}}
{d^2\chi_{0,0}^{\sop B}}\right|
+\left(\dfrac{d^2-1}{d^2}-\chi_{0,0}^{\sop B}\right),
\end{align}
which is only valid when $\Favg(\sop A,\sop I)\geq 2\Favg(\sop B,\sop
I)-1$, or, in the fidelity estimation context, when $\Favg(\sop E, \sop U)$ is close to
1~\footnote{In~\cite{MGJ+12}, $E$ is given as the minimum of two
  functions. However, in almost all realistic cases, this second function never
  minimizes, and in fact our bounds are always better than the second function,
  so we have left it out.}. There is no way to directly verify from
the experimental data that this requirement holds, but in order to
compare the bounds in Ref.~\cite{MGJ+12} with the bounds derived here, we
use \eq{bound} to bound region of validity of \eq{theirbound}. As
illustrated in \fig{comparebounds}, the bounds derived here are better
when $\Favg(\sop A\circ\sop B, \sop I)$ is close to 1, but are applicable to
the entire range of parameters without additional assumptions about
the maps involved.
  \begin{figure}[ht]
    \begin{center}
    \includegraphics[width=.45\textwidth]{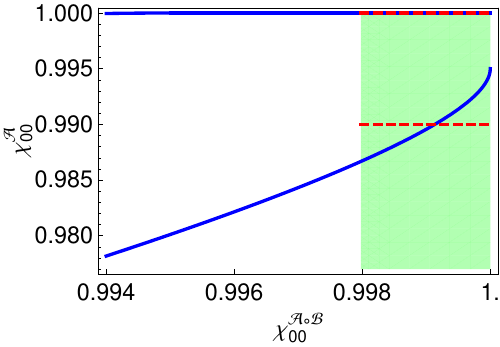}
        \end{center}
          \caption{Bounds on $\chi_{0,0}^{\sop A}$ versus
                $\chi_{0,0}^{\sop A\circ \sop B}$, when
                $\chi_{0,0}^{\sop B}$ is fixed at $0.995$. Our bounds
                are solid blue, while the bounds of Ref.~\cite{MGJ+12}
                are red dashed. The bounds of Ref.~\cite{MGJ+12} are
                valid in the green shaded region, while our bounds are
                valid for all values of $\chi_{0,0}^{\sop A\circ \sop
                  B}$.}
            \label{fig:comparebounds}
 \end{figure}

\subsection{Confidence Bounds on Fidelity Estimates}\label{sec:conf-bounds}
Here we show how to extract $F(\sop E,sop U)$ from the estimated
points $F_k(\sop E,\sop U)$ (the average fidelity of a length-$k$
RB sequence - see \ref{sec:RB}). We rigorously
bound the error and sampling complexity of this non-linear fit.

One can easily show that, using the Hoeffding bound, an estimate
$\widetilde{F}_k$ for $F_k(\sop E,\sop U)$ can be obtained such
that~\cite{MGE11}
\begin{align}
\label{eq:hoeff}
\Pr(|\widetilde{F}_k-F_k(\sop E,\sop U)|\geq \epsilon')\leq \delta'
\end{align}
with a number of samples $O\left ({1\over{\epsilon'}^2}\log{1\over
    \delta'}\right)$ that is independent of the number of qubits in
the system.  What we show here is that this allows for $p$ (and thus
$\overline{F}(\sop E,\sop U)$) to be estimated with a number of samples that also
scales well with some desired accuracy and confidence. In standard RB
experiments, $p$ is estimated by numerical fits to the
$\widetilde{F}_k$ with many different sequence lengths, but the
dependence of the error on the number of samples per sequence length
is difficult to analyse. Here we take a different approach that leads
to simple bounds on the accuracy and confidence.

Since $F_k(\sop E,\sop U)=A_0 p^k+B_0$, it is easy to see that
\begin{align}
p & = {F_2 - F_\infty\over F_1 - F_\infty},
\label{eqn:p-est}
\end{align}
and therefore, at least in principle, $p$ can be estimated by using
only sequences of length 1 and 2, along with a sequence long enough to
ensure $|A_0 p^k|\ll|B_0| $~\footnote{In practice this sequence length
  can be estimated roughly from the rough rate of decay in experiments
  without an accurate estimate for any of the model parameters.}, with
corresponding expectation denoted by $F_\infty$. Assuming each
$\widetilde{F}_i$ is estimated with accuracy $\epsilon'$ and
confidence $1-\delta'$, and that $0$ is not in the confidence interval
for $\widetilde{F}_1-\widetilde{F}_\infty$, it follows that the
estimate $\widetilde{p}$ for $p>0$ and $A_0>0$ satisfies
\begin{align}
{p - {2 \epsilon'\over A_0 p}\over 1+{2 \epsilon'\over A_0 p}}
\le \widetilde{p} \le
{p + {2 \epsilon'\over A_0 p}\over 1-{2 \epsilon'\over A_0 p}}
\label{eq:p-bounds}
\end{align}
with probability at least $1-3\delta'$ (similar expressions hold for
the cases negative $p$ or $A_0$, but, for simplicity, we focus on the
expressions for the positive case).  If $|A_0|$ or $|p|$ are small,
these bounds diverge, so it is important to test the data to exclude
these cases.

Note that $A_0$ is independent of the sequences being used so that one
can choose to estimate $A_0$ from a sequence with large $p$. Denoting
the $F_i$ estimates for those sequences as $\widetilde{F}'_i$, and
assuming the confidence interval for
$\widetilde{F}'_2-\widetilde{F}_\infty'$ does not include 0, $A_0$ is
bounded below via
\begin{align}
A_0 \ge {(\widetilde{F}'_1 - \widetilde{F}'_\infty-2\epsilon')^2\over \widetilde{F}'_2 - \widetilde{F}'_\infty+2\epsilon'}\equiv a,
\end{align}
with probability at least $1-3\delta'$. From \eq{hoeff},
\begin{align}
a p - 2 \epsilon' \le A_0 p - 2 \epsilon'\le
\widetilde{F}_1-\widetilde{F}_\infty,
\end{align}
and thus
\begin{align}
p \le {\widetilde{F}_1-\widetilde{F}_\infty + 2 \epsilon'\over a}
\end{align}
so that if one desires an accuracy $\epsilon$ for $\widetilde{p}$, whenever 
\begin{align}
{\widetilde{F}_1-\widetilde{F}_\infty + 2 \epsilon'\over a}\le \epsilon
\end{align}
one can set $\widetilde{p}=0$ thereby avoiding the divergent 
confidence intervals while still providing estimates with the desired
accuracy.

Similarly, from \eq{hoeff},
\begin{align}
\widetilde{F}_1-\widetilde{F}_\infty \le
A_0 p + 2 \epsilon',
\end{align}
so whenever
\begin{align}
{\widetilde{F}_1-\widetilde{F}_\infty + 2 \epsilon'\over a}\ge \epsilon,
\end{align}
it follows that
\begin{align}
a \epsilon - 4 \epsilon' \le A_0 p \le A_0,
\label{eq:A0p-bound}
\end{align}
so choosing $\epsilon'=4\epsilon^2a$ one can safely Taylor expand 
\eq{p-bounds} to first order in $\epsilon$ to obtain
\begin{align}
p - \epsilon - O(\epsilon^2)
\le \widetilde{p} \le
p + \epsilon + O(\epsilon^2),
\end{align}
with probability at least $1-\delta=1-6\delta'$, 
as desired, using $O\left ({1\over{\epsilon}^4}\log{6\over
    \delta}\right)$ samples.

This immediately gives that an estimate $\widetilde{F}$ for 
$\Favg(\sop E,\sop U)$ can be obtained such that
\begin{align}
\Pr(|\widetilde{F}_k-\Favg(\sop E,\sop U)|\geq \epsilon)\leq \delta
\end{align}
with $O\left ({1\over{\epsilon}^4}\log{1\over
    \delta}\right)$ samples.

\section{Summary and Outlook}
\label{sec:conclusion}

We have demonstrated that, using information from multiple RB
experiments, it is possible to reconstruct the unital part of any
completely-positive trace-preserving map in a way that is robust
against preparation and measurement errors, thereby avoiding some
forms of systematic errors that plague more traditional tomographic
reconstruction protocols.  The unital part of a map consists of the
vast majority of the parameters of that map, including all parameters
necessary to describe any deterministic unitary map, as well as any
random unitary map, such as dephasing with respect to any eigenbasis.

We also presented a robust procedure for bounding the average fidelity
to an arbitrary unitary, and show that this protocol is efficient for
a large class of unitaries outside of the Clifford group. The overhead
of the procedure depends on how the unitary is decomposed as a linear
combination of Clifford group unitary maps, and we give rigorous
bounds on the number of samples needed to achieve some desired
accuracy and confidence in the fidelity estimate.

The extension of these results to non-qubit systems remains an open
problem. In addition, the characterization of the non-unital part of a
map in a robust manner seems to present a larger challenge than the
characterization of the unital part. New techniques are needed to
access this important information.

\acknowledgements 

We thank Easwar Magesan for many illuminating discussions.  S.K. was
partially supported by the U.S. Department of Energy under cooperative
research agreement contract No. DE-FG02-05ER41360.  This research was
funded by the Office of the Director of National Intelligence (ODNI),
Intelligence Advanced Research Projects Activity (IARPA), through the
Army Research Office grant W911NF-10-1-0324. All statements of fact,
opinion or conclusions contained herein are those of the authors and
should not be construed as representing the official views or policies
of IARPA, the ODNI, or the U.S. Government.

%

\appendix

\section{Unital Maps and the Linear Span Of Unitary Maps}
\label{app:clifford-span}

The Pauli-Liouville representation is particularly convenient when
discussing the Clifford group of $n$-qubit unitary maps, because, in
this representation, such maps are {\em monomial
  matrices}~\cite{mvdn11a,mvdn11b}. In the particular case of qubits,
$\sop E_{ij}^\PL\in\{\pm1,0\}$ for a unitary in the Clifford
group. Given these facts, we can now straightforwardly prove the
result about the linear span of Clifford group maps on $n$
qubits. First, we need to prove a small result about Clifford group
unitaries.

\begin{claim} For any two $n$-qubit Pauli operators $\op P_{i\not=0}$ and
  $\op P_{j\not=0}$, there exists a Clifford group unitary $\op C$
  such that
  $\op C \op P_{i}\op C^\dagger=\op P_{j}$.\label{claim:PL-matrix-elements}
\end{claim}
\begin{proof}
  This claim shows that there are no subsets of non-identity
  multi-qubit Pauli operators that do not mix under the action of the
  Clifford group.  $\op P_i$ and $\op P_j$ can both be written as
  tensor products of single-qubit Pauli operators and identity
  operators $\I$, where in the tensor product of each there is at
  least one element that is not $\I$. Using local Clifford group
  unitaries one can take each non-identity element in each tensor
  product to the single-qubit Pauli operator $\op X$.  We call these
  new Pauli operators $\op P_i'$ and $\op P_j'$.  Now the problem is
  equivalent to finding Clifford group unitaries that take one tensor
  product of $\op X$ and $\I$ to another.

  Let $\CNOT_{k,l}$ denote the controlled-not unitary with qubit $k$
  as a control and qubit $l$ as a target. The $\CNOT$ is a well known
  unitary in the Clifford group with the property that $\CNOT_{k,l}\op
  X^{(k)}\CNOT_{k,l}=\op X^{(k)}\otimes \op
  X^{(l)}$, where we use $\op X^{(i)}$ to denote $\op X$
  acting on the $i^{\rm{th}}$ qubit. In this way one can
  increase or decrease the number of $\op X$ in the tensor
  product decomposition of $\op P_i'$ using unitary maps, as long as
  there is at least one $\op X$ in the tensor product.  This
  means that any tensor product of $\I$'s and $\op X$'s on
  $n$ qubits can be mapped to any other tensor product of $\I$ and
  $\op X$'s on $n$ qubits through the use of $\CNOT$
  unitaries---in particular, one can map $\op P_{i}'$ to $\op P_{j}'$.
\end{proof}

Now we can prove the intended result.

\cliffordspan*
\begin{proof}
  It suffices to show that any matrix element in the unital part of a
  map (in the Pauli-Liouville representation) can be written as a
  linear combination of Clifford group unitary maps.

  The Pauli-Liouville representation of unitaires in the $n$-qubit
  Clifford group are monomial matrices with non-zero entries equal to
  $\pm1$. For any given such unitary $\op C$, one can construct $4^n$
  orthogonal unitaries of the form $\op P_i\op C$, with corresponding
  $4^n$ mutually orthogonal Pauli-Liouville representation
  matrices. Pauli operators are diagonal in the Pauli-Liouville
  representation, so that for a fixed $\op C$, the Pauli-Liouville
  representations of all $\op P_i\op C$ have support in the same set
  of $4^n$ matrix elements as the Pauli-Liouville representation of
  $\op C$, and thus the values of any of these matrix elements for any
  map $\sop E$ can be recovered by collecting the Hilbert-Schmidt
  inner products between $\sop E^\PL$ and the Pauli-Louville
  representation of the map for the $\op P_i\op C$, i.e., $\tr \sop
  E (\sop P_i \sop C)^\dagger$. From
  Claim~\ref{claim:PL-matrix-elements}, one can choose a Clifford
  group unitary that has support on any particular matrix element in
  the unital block, therefore any unital matrix can be written as a
  linear combination of Clifford group unitary maps. Since Clifford
  group maps are unital, this concludes the proof.
\end{proof}

\section{Reconstruction of the unital part with imperfect operations}
\label{app:unital-recon}

In the main body of this paper we describe how RB allows for the
reconstruction of the unital parts of $\sop E\circ\sop N$ and $\sop
N$, where $\sop E$ is some quantum operations one would like to
characterize, and $\sop N$ is the error operation associated with each
of the randomizing operations. We now prove the result which allows
for the estimation of the unital part of $\sop E$ alone, given an
estimate of the unital part of $\sop N$.

\noisyunitalest*

\begin{proof}
  Any trace-preserving linear map $\sop A$ can be written in
  the Pauli-Liouville representation as
\begin{align}
    \sop A^\PL & = 
    \left(
      \begin{array}{cc}
        1               & \vec{0}^T\\
        \vec{t}_{\sop A} &\textbf{T}_{\sop A}
      \end{array}
    \right),
  \end{align}
  where, as discussed previsouly, the unital part is
\begin{align}
    \sop A'^\PL & = 
    \left(
      \begin{array}{cc}
        1               & \vec{0}^T\\
        \vec{0} &\textbf{T}_{\sop A}
      \end{array}
    \right).
  \end{align}
  The Pauli-Liouville representation of the composition of two
  trace-preserving linear maps $\sop A$ and $\sop B$ is given by
  the multiplication of the Pauli-Liouville representations, resulting in 
\begin{align}
    (\sop A\circ\sop B)^\PL & = 
    \left(
      \begin{array}{cc}
        1               & \vec{0}^T\\
        \vec{t}_{\sop A} + {\bf T}_{\sop A} \vec{t}_{\sop B} & {\bf T}_{\sop A} {\bf T}_{\sop B}
      \end{array}
    \right),
  \end{align}
  and thus
\begin{align}
    (\sop A\circ\sop B)'^\PL & = 
    \left(
      \begin{array}{cc}
        1               & \vec{0}^T\\
        \vec{0} & {\bf T}_{\sop A} {\bf T}_{\sop B}
      \end{array}
    \right),\\
    & = (\sop A)'^\PL (\sop B)'^\PL.
  \end{align}
  It follows immediatelly that 
\begin{align}
    (\sop A)'^\PL= (\sop A\circ\sop B)'^\PL [ (\sop B)'^\PL ] ^{-1},
  \end{align}
  if the inverse exists, and
\begin{align}
    \sop A'= (\sop A\circ\sop B)' \circ (\sop B') ^{-1},
  \end{align}
  by the Pauli-Liouville isomorphism. The lemma follows by setting
  $\sop A=\sop E\circ\sop N$ and $\sop B=\sop N$.
\end{proof}

\section{Complete-Positivity of the Projection of Single Qubit Operations onto the Unital Subspace}
\label{app:projection}
In this appendix, we prove that for a CPTP map $\sop E$ acting on a
single qubit, $\sop E'$, the unital part of $\sop E$ (see
\eq{nonunitalpart}), is always a CPTP map.

Recall that the Pauli-Liouville representation of a single qubit map
$\sop E$ may be written as
\begin{align}
\sop E^{\PL}=\left(
\begin{array}{cc}
1& \vec{0}\\
\vec{t}_{\sop E} &\textbf{T}_{\sop E}
\end{array}\right).
\end{align}
King and Ruskai \cite{KR01} show that there exist unitary maps $\sop U$
and $\sop V$ such that
\begin{align}
\sop U\breve{\sop E}\sop V=\sop E
\end{align}
where
\begin{align}
\breve{\sop E}^{\PL}=\left(
\begin{array}{cc}
1& \vec{0}\\
\vec{\tau} &\textbf{D}
\end{array}\right)=
\left(
\begin{array}{cccc}
1&0&0&0\\
\tau_1& \lambda_1&0&0\\
\tau_2&0& \lambda_2&0\\
\tau_3&0&0 &\lambda_3\\
\end{array}\right).
\end{align}

To prove $\sop E'$ is CPTP, we first show that $\breve{\sop E}'$ (the
projection of $\breve{\sop E}$ onto the unital block) is always CPTP,
and then we prove that if $\breve{\sop E}'$ is CPTP, $\sop E'$ is
CPTP.

\begin{lem}\label{lemm:hatE}
For single qubit operations, $\breve{\sop E}'$ is always CPTP
\end{lem}
\begin{proof}
Ruskai et al.~\cite{RSW02} prove that $\breve{\sop E}$ is CP if and only if
\begin{align}
&(\lambda_1+\lambda_2)^2\leq 
(1+\lambda_3)^2-\tau_3^2-(\tau_1^2+\tau_2^2)
\left(\frac{1+\lambda_3\pm \tau_3}{1-
\lambda_3\pm \tau_3}\right)
\label{eq:cond1}\\
&(\lambda_1-\lambda_2)^2\leq 
(1-\lambda_3)^2-\tau_3^2-(\tau_1^2+\tau_2^2)
\left(\frac{1-\lambda_3\pm \tau_3}{1+
\lambda_3\pm \tau_3}\right)
\label{eq:cond2}\\
&\left(1-(\lambda_1^2+\lambda_2^2+\lambda_3^2)-
(\tau_1^2+\tau_2^2+\tau_3^2)\right)^2\geq\nonumber\\ 
&4(\lambda_1^2(\tau_1^2+\lambda_2^2)+
\lambda_2^2(\tau_2^2+\lambda_3^2)+\lambda_3^2(\tau_3^2+
\lambda_1^2)-2\lambda_1\lambda_2\lambda_3)
\label{eq:cond3}
\end{align}
where in \eq{cond1} and \eq{cond2}, if $|\lambda_3|+|\tau_3|=1$, then
$\tau_1$ and $\tau_2$ must be 0 for the map to be CP. 

Notice that if these conditions are satisfied for a CPTP map $\breve{\sop
E}$, then they are also satisfied for the map $\breve{\sop E}^-$,
which is the same as $\breve{\sop E}$, except with
$\tau_1\rightarrow-\tau_1$, $\tau_2\rightarrow-\tau_2$,
$\tau_3\rightarrow-\tau_3$. Hence $\breve{\sop E}^-$ must also be
CPTP.

Now the convex combination of CPTP maps is also CPTP, so 
$1/2(\breve{\sop E}^-+\breve{\sop E})=\breve{\sop E}'$ is CPTP. 
\end{proof}

\unitalproject*
\begin{proof}
\lemm{hatE} shows that the projection of $\breve{\sop E}$ onto its unital 
part results in a CP map. So here we show this implies the projection of 
the map $\sop E$ onto its unital part results in a CP map.

Because $\sop U$ and $\sop V$ are unitaries, their Pauli-Liouville 
representations have the form 
\begin{align}
\sop U^{\PL}=\left(
\begin{array}{cc}
1& \vec{0}\\
\vec{0} &\textbf{U}
\end{array}\right),\text{ }
\sop V^{\PL}=\left(
\begin{array}{cc}
1& \vec{0}\\
\vec{0} &\textbf{V}
\end{array}\right).
\end{align}
So 
\begin{align}
\sop E^{\PL}=\sop U^{\PL}\breve{\sop E}^{\PL}\sop V^{\PL}=
\left(
\begin{array}{cc}
1& \vec{0}\\
\textbf{U}\vec{t} &\textbf{U}\textbf{D}\textbf{V}
\end{array}\right).
\end{align}
Now suppose $
\breve{\sop E}^{'\PL}=\left(
\begin{array}{cc}
1& \vec{0}\\
0 &\textbf{D}
\end{array}\right)$ is a valid CP map. Then
\begin{align}
\sop W^{\PL}=
\sop U^{\PL}\breve{\sop E}^{'\PL}\sop V^{\PL}=
\left(
\begin{array}{cc}
1& \vec{0}\\
0 &\textbf{U}\textbf{D}\textbf{V}
\end{array}\right)
\end{align}
 is also a valid CPTP map because the composition of valid quantum
maps is always a valid quantum map. However, by \eq{nonunitalpart}
 $\sop W$ is equal to $\sop E'$, so the unital part of a single qubit map
is always CPTP.
\end{proof}

\section{Bounds on Fidelity}
\label{app:bound}
Recall that for an operation $\sop E$, the $\chi$-matrix representation is
\begin{align}
\sop E(\op\rho)=\sum_{i,j}{\chi}_{i,j}^{\sop E}\op{P}_i\op\rho\op{P}_j.
\end{align}
Due to complete positivity constraints $\chi$ 
matrix elements satisfy 
\begin{align}
\label{eq:pos_constr}
\chi_{i,j}^{\sop E}\leq\sqrt{\chi_{i,i}^{\sop E}\chi_{j,j}^{\sop E}}.
\end{align}

Composing two maps, their $\chi$-matrix represenations compose as
\begin{align}
\sop A\circ\sop B(\op \rho)=\sum_{m,n,k,j}\chi_{m,n}
^{\sop A}\chi_{k,j}^{\sop B}\op P_m\op P_k\op\rho \op P_j\op P_n.
\end{align}
Let $\sigma_i(m)$ be the index such that $\op P_{\sigma_i(m)}\op P_m=\op P_i$. Then 
using the fact that the absolute value is greater than the real or imaginary
 parts of a complex number, we obtain
\begin{align}
\label{eq:boundfirststep}
\chi_{i,i}^{\sop A\circ \sop B}&=
\chi_{i,i}^{\sop A}\chi_{0,0}^{\sop B}\pm\left(2\sum_{m\neq0}\left|\chi^{\sop
 A}_{\sigma_i(m),i}\right| \Big|\chi^{\sop B}_{m,0}\Big|\right.\nonumber\\
&\left.+\sum_{m,n\neq0}\left|\chi^{\sop A}_{\sigma_i(m),\sigma_i(n)}\Big| 
\Big|\chi^{\sop B}_{m,n}\right|\right).
\end{align}
Looking at the term $\sum_{m\neq0}|\chi^{\sop A}_{\sigma_i(m),i}| 
|\chi^{\sop B}_{m,0}|$ and using \eq{pos_constr} and the Cauchy-Schwarz inequality, we have
\begin{align}
\sum_{m\neq0}&\Big|\chi^{\sop A}_{\sigma_i(m),i}\Big| \Big|\chi^{\sop B}_{m,0}\Big|\leq
 \sqrt{\sum_{m\neq 0}|\chi^{\sop A}_{\sigma_i(m),i}|^2\sum_{m\neq0}
|\chi^{\sop B}_{m,0}|^2}\nonumber\\
&\leq\sqrt{\sum_{m\neq0}\chi^{\sop A}_{ \sigma_i(m),
 \sigma_i(m)}\chi^{\sop A}_{i,i}\sum_{m
\neq0}\chi^{\sop B}_{m,m}\chi^{\sop B}_{0,0}}\nonumber\\
&=\sqrt{(1-\chi^{\sop A}_{i,i})\chi^{\sop A}_{i,i}(1-\chi^{\sop
 B}_{0,0})\chi^{\sop B}_{0,0}}
\end{align}
Similarly, the term $\sum_{m,n\neq0}\left|\chi^{\sop A}_{\sigma_i(m),
\sigma_i(n)}\Big| 
\Big|\chi^{\sop B}_{m,n}\right|$ gives
\begin{align}
&\sum_{m,n\neq0}\left|\chi^{\sop A}_{\sigma_i(m),\sigma_i(n)}\Big| 
\Big|\chi_{m,n}\right|\nonumber\\
&\leq\sqrt{\sum_{m, n\neq 0}|\chi^{\sop A}_{\sigma_i(m),\sigma_i(n)}|^2\sum_{m,
 n\neq 0}|\chi^{\sop B}_{m,n}|^2}\nonumber\\
&\leq\sqrt{\sum_{m, n\neq 0}\chi^{\sop A}_{\sigma_i(m),
\sigma_i(m)}\chi^{\sop A}_{\sigma_i(n),\sigma_i(n)}\sum_{m, n\neq 
0}\chi^{\sop B}_{m,m}\chi^{\sop B}_{n,n}}\nonumber\\
&=(1-\chi^{\sop A}_{i,i})(1-\chi^{\sop B}_{0,0})
\end{align}

So we have
\begin{align}
\chi^{\sop A\circ\sop B}_{i,i}&=\chi^{\sop A}_{i,i}\chi^{\sop B}_{0,0}\pm2\sqrt{(1-\chi^{\sop A}_{i,i})\chi^{\sop A}_{i,i}
(1-\chi^{\sop B}_{0,0})\chi^{\sop B}_{0,0}}\nonumber\\
&\pm(1-\chi^{\sop A}_{i,i})(1-\chi^{\sop B}_{0,0}).
\end{align}
Setting $i=0$ gives the desired result.

To see why the lower bound is in general not tight, consider
\eq{boundfirststep}. For the lower bound on $\chi_{i,i}^{\sop
  A\circ\sop B}$ we take all of the terms of the form $\chi^{\sop
  A}_{\sigma_i(m), \sigma_i(n)} \chi^{\sop B}_{m,n}$ and replace them
with $-|\chi^{\sop A}_{\sigma_i(m),\sigma_i(n)}||\chi^{\sop B}_{m,n}|$
because many of these terms have unknown phases, which in the worst
case can have value $-1$. However, when $m=n$, because $\chi$ is
positive semidefinite, we get terms of the form $\chi^{\sop
  A}_{\sigma_i(m), \sigma_i(m)} \chi^{\sop B}_{m,m}=|\chi^{\sop
  A}_{\sigma_i(m), \sigma_i(m)}|| \chi^{\sop B}_{m,m}|$, so we are
subtracting terms which should actually be added. However, there is no
way to address this issue without obtaining more information about the
$\chi$ matrix.

\section{Ordering Of Error Maps}\label{app:gauge}

Throughout this paper, we chose to describe the noisy maps as the
composition of the ideal map and some error map (applied in that
order). That is, the noisy implementation of the map $\sop C_i$ is
expressed as
\begin{align}
\sop N_i\circ\sop C_i
\end{align}
where $\sop N_i$ is the error map and $\sop C_i$ is the ideal Clifford
map.  This choice can be made without loss of generality, and has no
effect on experimental observations.  That is, we could instead
express the implementation of the Clifford $\sop C_i$ as
\begin{align}
\sop C_i\circ\sop N_i^*
\end{align}
where $\sop N_i^*$ is, in general, a different error map. The average
fidelity of these noisy maps to $\sop C_i$ is the same, or,
equivalently, $\sop N_i$ and $\sop N_i^*$ have the same average
fidelity to the identity. However, in general $\sop N_i\not=\sop
N_i^*$. Other process metrics are immune to this problem because they
take the error to be additive rather than multiplicative, and so there
is no ordering choice to be made or imposed.

If all error maps are identical for either of the conventional choices
($\sop N_i=\sop N$ or $\sop N_i^*=\sop N^*$) then \eqref{eq:model}
holds, and small deviations from these cases lead to perturbative
corrections that generalize the results in Ref.~\cite{MGE11,MGE12}. If
the error maps are close to the identity, both perturbative models are
likely to be valid, so $\sop N\approx\sop N^*$---the question of which
convention is used become immaterial. However, if either $\sop N$ or
$\sop N^*$ is far from the identity, low order perturbative expansions
may not be valid for one of the conventions. Individual RB fits cannot
differentiate between these two cases, and the bounds used to isolate
the error in $\sop E$ from $\sop N$ or $\sop N^*$ do not depend on
this conventional choise, so as long as \eqref{eq:model} holds for
{\em some} separation of the error and ideal channel.

A problem arises when one attempts to use
Lemma~\ref{lemm:noisy-unital-est}, as, unless $\sop N\approx\sop N^*$,
the choice of conventions becomes important. The physical regime
where, e.g., $\sop N_i\approx\sop N$ is precisely the regime where
$\sop N_i\approx \sop N_i^*\approx\sop I$, and so this not not likely
to be a problem in practice---within the accuracy of the perturbative
expansions to \eqref{eq:model}, the inversion in
Lemma~\ref{lemm:noisy-unital-est} will be valid, as would a similar
inversion taking the error map to ocurr before the ideal map.

In the more general formal cases where, e.g., $\sop N_i\approx\sop N$
but the $\sop N_i^*$ are very different from each other, there appears
to be no way to choose the appropriate convention from individual
observations. It may simply be the case that the $\sop E'$
reconstruction via Lemma~\ref{lemm:noisy-unital-est} using one
convention is highly unphysical, while the other is physical,
indicating which convention should be used. In the absence of this
indication of systematic errors, however, one should report both
reconstructions or simply choose the worst of the two.
\end{document}